\documentclass[journal]{IEEEtran}

\usepackage[T1]{fontenc}
\usepackage[utf8]{inputenc}

\usepackage[cmex10]{amsmath}
\usepackage{graphics}
\usepackage{dsfont}
\usepackage{amssymb}
\usepackage{graphicx}
\usepackage{mathrsfs}
\usepackage{units}
\usepackage{tikz}
\usepackage{verbatim}
\usepackage{pgfplots}
\usepackage[colorlinks=true,linkcolor=black,citecolor=black,plainpages=false,pdfpagelabels]{hyperref}
\hypersetup{pdftitle={A decoupling approach to classical data transmission over quantum channels}}
\usetikzlibrary{%
  arrows,%
  calc,%
  fit,%
  patterns,%
  plotmarks,%
  shapes.geometric,%
  shapes.misc,%
  shapes.symbols,%
  shapes.arrows,%
  shapes.callouts,%
  backgrounds,%
  chains,%
  topaths,%
  trees,%
  petri,%
  mindmap,%
  matrix,%
  folding,%
  fadings,%
  through,%
  positioning,%
  scopes,%
  decorations.fractals,%
  decorations.shapes,%
  decorations.text,%
  decorations.pathmorphing,%
  decorations.pathreplacing,%
  decorations.footprints,%
  decorations.markings,%
  shadows}
\usepackage{tikz}
\usepackage{float}
\floatstyle{boxed} 
\restylefloat{figure}

\newtheorem{theorem}{Theorem}[section]
\newtheorem{lemma}[theorem]{Lemma}
\newtheorem{proposition}[theorem]{Proposition}
\newtheorem{corollary}[theorem]{Corollary}
\newenvironment{proof}[1][Proof]{\begin{trivlist}
\item[\hskip \labelsep {\bfseries #1}]}{\end{trivlist}}
\newenvironment{definition}[1][Definition]{\begin{trivlist}
\item[\hskip \labelsep {\bfseries #1}]}{\end{trivlist}}

\newcommand{\qed}{\nobreak \ifvmode \relax \else
      \ifdim\lastskip<1.5em \hskip-\lastskip
      \hskip1.5em plus0em minus0.5em \fi \nobreak
      \vrule height0.75em width0.5em depth0.25em\fi}

\newcommand{\abs}[1]{\ensuremath{|#1|}}

\newcommand{\norm}[2]{\ensuremath{|\!|#1|\!|_{#2}}}

\newcommand{\Norm}[2]{\ensuremath{\left|\!\left|#1\right|\!\right|_{#2}}}

\newcommand{\tr}{\textnormal{tr}}
\newcommand{\trace}[1]{\ensuremath{\tr (#1)}}
\newcommand{\Trace}[1]{\ensuremath{\tr \left( #1 \right)}}

\newcommand{\ptr}[1]{\textnormal{tr}_{\textnormal{ #1}}}
\newcommand{\ptrace}[2]{\ensuremath{\ptr{#1} (#2)}}
\newcommand{\Ptrace}[2]{\ensuremath{\ptr{#1} \left(#2\right)}}

\newcommand{\idx}[2]{{#1}_{#2}}

\newcommand{\ket}[1]{| #1 \rangle}
\newcommand{\keti}[2]{| #1 \rangle_{\textnormal{ #2}}}
\newcommand{\bra}[1]{\langle #1 |}

\newcommand{\bracket}[3]{\langle #1 | #2 | #3 \rangle}

\newcommand{\proj}[2]{| #1 \rangle\!\langle #2 |}
\newcommand{\proji}[3]{| #1 \rangle\!\langle #2 |_{\textnormal{#3}}}

\newcommand{\kron}{\otimes}

\newcommand{\eps}{\varepsilon}

\newcommand{\id}{\ensuremath{\mathds{1}}}
\newcommand{\idi}[1]{\ensuremath{\mathds{1}_{#1}}}

\newcommand{\idA}{\idi{A}}

\newcommand{\opid}{\ensuremath{\mathcal{I}}}

\newcommand{\linops}[1]{\ensuremath{\mathcal{L}(#1)}}
\newcommand{\hermops}[1]{\ensuremath{\mathcal{L}^\dagger(#1)}}
\newcommand{\posops}[1]{\ensuremath{\mathcal{P}(#1)}}

\newcommand{\states}[1]{\ensuremath{\mathcal{S}(#1)}}
\newcommand{\normstates}[1]{\ensuremath{\mathcal{S}_{=}(#1)}}
\newcommand{\subnormstates}[1]{\ensuremath{\mathcal{S}_{\leq}(#1)}}

\newcommand{\cC}{\mathcal{C}}
\newcommand{\cD}{\mathcal{D}}
\newcommand{\cE}{\mathcal{E}}

\newcommand{\cH}{\mathcal{H}}
\newcommand{\cI}{\mathcal{I}}

\newcommand{\cN}{\mathcal{N}}

\newcommand{\cT}{\mathcal{T}}

\newcommand{\rhot}{\ensuremath{\tilde{\rho}}}

\newcommand{\rhoA}{\ensuremath{\idx{\rho}{A}}}
\newcommand{\rhoB}{\ensuremath{\idx{\rho}{B}}}

\newcommand{\rhoAB}{\ensuremath{\idx{\rho}{AB}}}

\newcommand{\rhoAR}{\ensuremath{\idx{\rho}{AR}}}

\newcommand{\sigmaA}{\ensuremath{\idx{\sigma}{A}}}
\newcommand{\sigmaB}{\ensuremath{\idx{\sigma}{B}}}

\newcommand{\zetaR}{\ensuremath{\idx{\zeta}{R}}}

\newcommand{\cTAE}{\ensuremath{\idx{\cT}{A\rightarrow B}}}
\newcommand{\HA}{\ensuremath{\idx{\cH}{A}}}
\newcommand{\HB}{\ensuremath{\idx{\cH}{B}}}

\newcommand{\chh}[5]{\ensuremath{H_{#1}^{#2}(\textnormal{#3}|\textnormal{#4})_{#5}}}

\newcommand{\chmin}[3]{\chh{\textnormal{min}}{}{#1}{#2}{#3}}
\newcommand{\chmineps}[3]{\chh{\textnormal{min}}{\eps}{#1}{#2}{#3}}
\newcommand{\chmineeps}[4]{\chh{\textnormal{min}}{#1}{#2}{#3}{#4}}

\newcommand{\chmaxeps}[3]{\chh{\textnormal{max}}{\eps}{#1}{#2}{#3}}

\newcommand{\EP}[1]{\ensuremath{\underset{\mathbb{P}}{\textnormal{\large{$\mathbb{E}$}}}\:#1}}

\newcommand{\ident}{\mathds{1}}
\newcommand{\Hmin}{H_{\min}}
\DeclareMathOperator{\cl}{cl}
\DeclareMathOperator{\Span}{span}

\newcommand{\mfX}{\mathfrak{X}}

\newcommand{\mbE}{\mathbb{E}}

\newcommand{\mbP}{\mathbb{P}}

\newcommand{\commentout}[1]{}

\begin{document}
%
\title{A decoupling approach to classical data transmission over quantum channels}
%
%
%

\author{Frédéric~Dupuis,
        Oleg~Szehr,
        and~Marco~Tomamichel
\thanks{
F. Dupuis is with ETH Z\"urich, 
O. Szehr is with TU M\"unchen
and M. Tomamichel is with CQT, National University of Singapore.}%
}

%
%

%

\maketitle

\begin{abstract}
Most coding theorems in quantum Shannon theory can be proven using the decoupling technique: to send data through a channel, one guarantees that the environment gets no information about it; Uhlmann's theorem then ensures that the receiver must be able to decode. While a wide range of problems can be solved this way, one of the most basic coding problems remains impervious to a direct application of this method: sending classical information through a quantum channel. We will show that this problem can, in fact, be solved using decoupling ideas, specifically by proving a ``dequantizing'' theorem, which ensures that the environment is only classically correlated with the sent data.

Our techniques naturally yield a generalization of the Holevo-Schumacher-Westmoreland Theorem to the one-shot scenario, where a quantum channel can be applied only once.
\end{abstract}


%
\section{Introduction}
\label{intro}
One of the most fruitful ideas that arose in quantum Shannon theory in the past few years is that of \emph{decoupling}: the fact that, in quantum mechanics, the absence of correlations between two systems implies perfect correlations of those two systems with a third one. More precisely, the core idea is as follows: suppose that we have a tripartite pure state $\ket{\rho}_{ABC}$, and that we know that the reduced state on $AB$ is a product state, i.e. $\tr_C[\proj{\rho}{\rho}] = \rho_A \otimes \rho_B$. Then, we know from the unitary equivalence of purifications that there exists a partial isometry $V_{C \rightarrow C_A C_B}$ with the property that $V \ket{\rho} = \ket{\psi}_{AC_A} \otimes \ket{\varphi}_{B C_B}$. In other words, if $A$ and $B$ are completely uncorrelated, then $C$ contains perfect correlations with both $A$ and $B$. Furthermore, this observation remains true if the state on $A$ and $B$ is only close to a product state, as can be shown via Uhlmann's theorem \cite{uhlmann}.

This observation can be used to prove coding theorems for quantum Shannon theory problems. To see this, suppose that we have a channel $\cT_{A \rightarrow B}$, with a Stinespring dilation $(U_{\mathcal{T}})_{A\rightarrow BE}$, and that we want to use this channel to send quantum information from Alice (who has access to the input $A$) to Bob (who receives the output system $B$). Let $\psi_M$ (with purification $\ket{\psi}_{MR}$) be the state of the message Alice wants to send, and let $W_{M \rightarrow A}$ be the encoding isometry she uses to map her state to the channel input. After encoding the state and sending it through the channel, we have $\varphi_{RBE} := U_{\mathcal{T}} W\psi W^\dagger U_{\mathcal{T}}^\dagger$. Now, suppose that the encoding operation is such that $\varphi_{RE} = \psi_R \otimes \varphi_E$. Then, the argument in the previous paragraph tells us that there exists an isometry $V_{B \rightarrow M E'}$ such that $V \varphi_{RBE} V^\dagger = \psi_{MR} \otimes \xi_{EE'}$ for some state $\xi$. If we then trace out $EE'$, we see that $V$ acted as a decoder to recover the initial state $\psi_{MR}$. One can also show that the condition that $R$ and $E$ be decoupled is not only sufficient but necessary in order to be able to transmit arbitrary quantum information. This simplifies our task as information theorists: as long as we can design an encoder $W$ that ensures that this decoupling condition is fulfilled, we know that a decoder must exist, and do not need to explicitly construct it. Furthermore, our aim becomes to \emph{destroy} correlations rather than to ensure their presence, which seems to be a rather less delicate task at first glance.

To enforce the decoupling condition, a number of \emph{decoupling theorems} have arisen \cite{state-merging,fqsw,fred-these,dbwr10}. The version from \cite{fred-these,dbwr10}, whose approach we will broadly follow here, goes as follows. Let $\mathcal{\bar{T}}_{A \rightarrow E}$ be a complementary channel for $\cT_{A\rightarrow B}$ and let $\rho_{AR}$ be a quantum state. We consider the state $(\bar{\mathcal{T}} \otimes \mathcal{I}_R) \big( (U_A \otimes \ident_R) \rho_{AR} (U_A^{\dagger} \otimes \ident_R) \big)$ on $ER$, where $U_A$ is chosen randomly according to the Haar measure on $\mathbb{U}(A)$. It turns out that this state is decoupled (i.e. that it is close to $\bar{\mathcal{T}}(\ident/d_A) \otimes \rho_R$ in trace distance) if the state and the channel fulfill a certain entropic criterion, namely that $\Hmin^{\varepsilon}(A|R)_{\rho} + \Hmin^{\varepsilon}(A'|E)_{\tau} \gtrsim 0$ (these \emph{smooth min-entropies} will be defined in the next section). Roughly speaking, the first term measures how hard the state $\rho_{AR}$ is to decouple, and the second term measures the ``decoupling power'' of the channel $\bar{\mathcal{T}}$; if the decoupling power of the channel exceeds the difficulty of decoupling the state, decoupling does indeed happen.

By appropriately applying the outlined procedure, one can get a variety of coding theorems. This general approach has now become a staple of quantum Shannon theory, and has been used in quantum state merging \cite{state-merging}, state transfer (also known as ``Fully Quantum Slepian-Wolf'') \cite{fqsw}, for sending quantum information through quantum channels \cite{lsd-decoupling}, for quantum broadcast channels \cite{dhl09}, quantum channels with side-information \cite{gpquantique}, among other examples.

The common point in all of the previous papers is that they use this argument to send \emph{quantum} information. For sending classical information, on the other hand, the argument does not work directly. The reason for this is that if one sends classical information, the channel environment (the system $E$ above) can also receive a copy of the message without impairing the protocol. However, it turns out that for the protocol to work, $E$  can \emph{only} share classical correlations with the message; in particular, $E$ cannot contain any phase information about the message, otherwise Bob cannot decode. Hence, while the vast majority of quantum Shannon theory can now be done using decoupling methods, classical coding over quantum channels, the so-called Holevo-Schumacher-Westmoreland theorem  (HSW theorem) \cite{holevo98,SW97}, remains a notable outlier. The purpose of this paper is to close this gap and provide a decoupling proof, based on the above argument, of the HSW theorem.

The results presented here have a somewhat similar flavor to those presented in \cite{rr11}, but a rather different emphasis. In both papers, the idea that the environment cannot have information about the phase of the classical message arises as a central theme. In \cite{rr11}, this occurs in the context of using complementary bases to get coding theorems from privacy amplification and information reconciliation, whereas here it arises as a natural analog of the concept of decoupling.

The paper will be structured as follows. Section \ref{prel} will explain the notation and basic concepts needed for this paper, Section \ref{sec:dequantizing-thm} will give a \emph{dequantizing theorem}, which will be the analog of the decoupling theorem that we will need for the classical case, and Section \ref{sec:proof-hsw} will show how to use it to derive coding theorems for sending classical information over quantum channels. Finally, we discuss the results in Section \ref{sec:discussion}.

\section{Preliminaries and Notation}
\label{prel}

\subsection{Quantum States and Maps}
\label{prel:not}
Let $\cH$ be a finite dimensional, complex Hilbert space. The set of linear operators on $\cH$ will be denoted by $\linops{\cH}$, the set of Hermitian operators by $\hermops{\cH}$ and the set of positive-semidefinite operators is given by $\posops{\cH}$. The set of quantum states is given by $\normstates{\cH} := \{\rho\in\posops{\cH} \mid \tr\, \rho = 1\}$ and the set of subnormalized quantum states is $\subnormstates{\cH} := \{\rho\in\posops{\cH} \mid \tr\,\rho \leq 1\}$. A subscript letter following some mathematical object denotes the physical system to which it belongs. However, when it is clear which systems are described we might drop the subscripts to shorten the notation. Given two physical systems $A$ and $B$, the joint bipartite system $AB$ is represented by a tensor product space $\idx{\cH}{A}\otimes\idx{\cH}{B}=:\idx{\cH}{AB}$.

We will denote by $\idi{A}$ the identity operator on $\HA$ and by $\idx{\pi}{A}:= \idi{A} / \idx{d}{A}$ the completely mixed state on $A$, where $d_A = \dim \cH_A$. For $\idx{d}{A}\geq\idx{d}{B}$ the states $\idx{T}{AB}:=\frac{1}{\idx{d}{B}}\sum_{i}^{\idx{d}{B}}{\proji{i}{i}{A}\otimes\proji{i}{i}{B}}$ and $\idx{\Phi}{AB}:=\frac{1}{\idx{d}{B}}\sum_{i,j}^{\idx{d}{B}}{\proji{i}{j}{A}\otimes\proji{i}{j}{B}}$ in $\normstates{\cH_{AB}}$ represent maximal classical and, respectively, quantum correlations between the systems $A$ and $B$.

Suppose $\ket{\psi}_{AB}$ is a pure state of the bipartite system $AB$ (i.e. the system is in the state $\psi_{AB}=\proji{\psi}{\psi}{AB}$) and $d_A\geq d_B$. Then there exist lists of orthonormal vectors $\{\ket{i}_A\}_{i=1,...,d_B}\in\cH_A$ and $\{\ket{i}_B\}_{i=1,...,d_B}\in\cH_B$ such that $\ket{\psi}_{AB}=\sum_i\lambda_i\ket{i}_A\ket{i}_B$, where $\lambda_i\geq0$ and $\sum_i\lambda_i^2=1$ \cite{NC2000}. The corresponding basis $\{\ket{i}\}_{i=1,...,d_B}$ is called \emph{Schmidt basis} and the numbers $\lambda_i$ are \emph{Schmidt coefficients}.

A quantum state $\rho_{AB}\in\subnormstates{\cH_{AB}}$ is said to be classical with respect to a fixed basis $\{\ket{i}\}_{i=1,...,d_A}$ of $\cH_A$ if $\rho_{AB}\in\textnormal{span}_{\mathbb{R}}\{\proj{i}{i}_{i=1,...,d_A}\}\otimes\hermops{\cH_B}$. If in addition $\rho_{AB}$ is not classical on $\cH_B$ we call it a hybrid classical-quantum or shortly \emph{CQ-state}. 
Moreover, we call a state $\rho_{XX'B} \in \mathcal{S}_{\leq}
(\mathcal{H}_{XX'B})$ \emph{coherent classical} on $X$ and $X'$ if it commutes with the projector $P_{XX'} = \sum_x \proj{x}{x}_{X} \otimes \proj{x}{x}_{X'}$.

Linear maps from $\linops{\idx{\cH}{A}}$ to $\linops{\idx{\cH}{B}}$ will be denoted by calligraphic letters, e.g.~$\idx{\cT}{A\ensuremath{\rightarrow}B}\in\textnormal{Hom}(\linops{\HA},\linops{\HB})$. Quantum operations are in one-to-one correspondence with trace preserving  completely positive maps (TPCPMs). The TPCPM we will encounter most often is the partial trace (over the system $B$), denoted $\ptrace{B}{\cdot}$, which is defined to be the adjoint mapping of $\idx{\mathcal{T}}{A\rightarrow AB}(\idx{\xi}{A}) = \idx{\xi}{A} \otimes\idi{B}$ for $\idx{\xi}{A}\in\hermops{\HA}$ with respect to the Schmidt scalar product 
$\langle A,B\rangle \ := \ \tr(A^{\dagger}B)$. This means $\tr( (\idx{\xi}{A}\otimes\idi{B})\idx{\zeta}{AB}) = \tr(\idx{\xi}{A}\: \ptrace{B}{\idx{\zeta}{AB}})$ for any $\idx{\zeta}{AB}\in\hermops{\idx{\cH}{AB}}$.
Given a bipartite state $\idx{\xi}{AB}$, we write $\idx{\xi}{A} := \ptr{B}{\idx{\xi}{AB}}$ for the reduced density operator on $A$ and $\idx{\xi}{B} := \ptr{A} \xi_{AB}$, respectively, on $B$. If $\idx{\xi}{AB}$ is pure, we call $\keti{\xi}{AB}$ a purification for $\idx{\xi}{A}$ and $\idx{\xi}{B}$.

The map $\cC(\cdot)_{A}=\sum_i \proji{i}{i}{A}(\cdot)\proji{i}{i}{A}$ classicalizes an arbitrary density operator on $A$ by removing all off-diagonal elements. When $\cC$ is applied to part of a bipartite state $\rho_{AB}$, we get the CQ state $\rho_{AB}^{\cl} := (\cC_A \otimes \cI_B)(\rho_{AB})$.  Here, $\idx{\opid}{B}$ denotes the operator identity on $B$, which we will only write explicitly if it is not clear from the context.

The Choi-Jamiołkowski representation~\cite{cj-choi,cj-jamiolkowski} of $\cTAE\in\textnormal{Hom}(\linops{\HA},\linops{\HB})$ is given by the operator ${\idx{\omega}{A'B}:=(\cTAE\otimes\idx{\opid}{A'})(\idx{\Phi}{AA'})}$, where $\cH_{A'}$ is a copy of $\cH_{A}$. We say that $\cTAE\in\textnormal{Hom}(\linops{\HA},\linops{\HB})$ has classical-quantum (CQ) structure if its Choi-Jamiołkowski representation is a CQ-state. For a map $\cT$ with Choi-Jamiołkowski $\omega_{A'B}$ we define the map $\cT^{\cl}$ to be the unique map whose Choi-Jamiołkowski representation is $\omega_{A'B}^{\cl} = \cC_{A'}(\omega_{A'B})$.


In our context it will also be important to purify quantum channels. Given a TPCPM $\cTAE\in\textnormal{Hom}(\linops{\HA},\linops{\HB})$, we define the unitary $(U_{\cT})_{A \rightarrow BE}$ to be any particular Stinespring dilation of $\cT$. The purifying system $E$ will be called the \emph{environment} of the channel. For a channel $\cT_{A \rightarrow B}$, we define the complementary channel $\bar{\cT}_{A \rightarrow E}: X \mapsto \tr_B[U_{\cT} X (U_{\cT})^{\dagger}]$ to be the channel to the environment. 

The purification of a CQ-channel $\cT$ with $\cT(\xi_A)=\sum_i{\Trace{\proji{i}{i}{A}\xi_A}}\rho_B^{[i]}$ is given by $U^{\cT}\keti{i}{A}=\keti{i}{X}\otimes\keti{\rho^{[i]}}{BE\rq{}}$. Thus, the environment of such a channel can conceptually be split into two parts: a register $X$, which contains a copy of the input to the channel and a system $E\rq{}$ which stems from the purification of the operators $\rho_B^{[i]}$. See Figure \ref{fig:purification-cq-channel} for an illustration of this. The Choi-Jamiołkowski representation of a complementary channel of a CQ-channel can be written as $\omega_{A\rq{}E\rq{}X}$, where the systems $X$ and $A\rq{}$ are classically coherent. Furthermore, we will frequently be considering channels that are complementary to CQ channels; we will call such channels ``complementary CQ channels''.

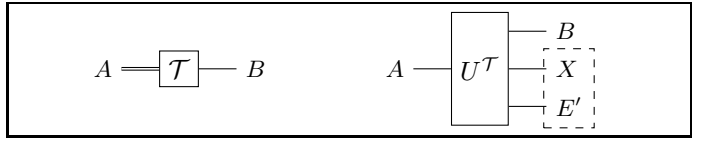
\begin{figure}
	\centering
		\begin{tikzpicture}
			\tikzstyle{gate} = [fill=white, draw]
			\tikzstyle{syslabel} = [font=\small]
			\node[gate] at (-4,0) (origchan) {$\mathcal{T}$};
			\draw[double] (origchan.west) -- ++(-.5,0) node[syslabel,left] {$A$};
			\draw (origchan.east) -- ++(.5,0) node[syslabel,right] {$B$};
			\node[gate,minimum height=1.5cm] at (0,0) (channel) {$U^{\mathcal{T}}$};
			\draw (channel.west) -- ++(-.5cm,0) node[syslabel,left] {$A$};
			\draw (channel.east)++(0,.5) -- ++(.5,0) node[syslabel,right] {$B$};
			\draw (channel.east) -- ++(.5,0) node[syslabel,right] (Xlabel) {$X$};
			\draw (channel.east)++(0,-.5) -- ++(.5,0) node[syslabel,right] (Elabel) {$E\rq{}$};
			\node[draw,dashed,fit=(Xlabel) (Elabel), inner sep=1pt] (dashedbox) {};
		\end{tikzpicture}
\caption{Diagram illustrating the purification of a CQ-channel. The environment (depicted with a dashed box) of a CQ-channel can be split into two parts, a register $X$ that contains a copy of the input and a system $E\rq{}$. }
\label{fig:purification-cq-channel}
\end{figure}

The swap operator $F_{AA'}$ acting on the bipartite space $\cH_{AA'}$ is given by $F_{AA'} := \sum_{i,j} \proji{i}{j}{A} \otimes \proji{j}{i}{A'}$.
It is easy to verify that for any $M_A$, $N_{A'}\in \linops{\cH_A}$ the swap operator satisfies
$\tr(M_A N_{A'}) \ = \ \tr\big((M_A \otimes N_{A'})F_{AA'}\big)$.

For any operator in $\xi_A\in\linops{\cH_A}$ we denote by $\norm{\xi_A}{1}$ and  $\norm{\xi_A}{2}$ the Schatten 1 and 2-norms of $\xi_A$, respectively. These norms are unitarily invariant and satisfy $\norm{\xi_A}{2}\leq\norm{\xi_A}{1}\leq\sqrt{d_A}\norm{\xi_A}{2}$.
The metric induced on $\linops{\cH}$ via the Schatten 1-norm is $D(\rho,\sigma):= \| \rho - \sigma \|_1$.
Another measure of closeness between states on $\posops{\cH}$ is the fidelity, $F(\rho,\sigma) := \|\sqrt{\rho} \sqrt{\sigma} \|_1$.

\subsection{Permutation operators}

The symmetric group $S_d$ is the set of all bijective maps of $\{1,...,d\}$ to itself together with the concatenation of maps as the group multiplication. Elements $\pi$ of $S_d$ are called permutations.
Let $\cH$ be a Hilbert space together with a fixed basis $\{\ket{i}\}_{i=1,...,d}$. For $\pi\in S_d$, we define the permutation operator $P(\pi)$ on $\cH$ such that $P(\pi) \ket{i} = \ket{\pi(i)}$. The group of all such matrices will be denoted by $\mathbb{P}$. Typically in this paper $\{ \ket{i} \}_{i=1,...,d}$ will be the Schmidt basis of a given density matrix. The above permutation matrices then act by reordering the elements of this basis.

Given a random variable $X:\mathbb{P}\rightarrow \Omega$ ($\Omega$ some measurable space),
we shall write $\mathbb{E}_{\mathbb{P}}{[X]}:=\frac{1}{d!}\sum_{P\in\mathbb{P}} X(P)$ for the expectation value of $X$ with respect to the uniform probability distribution on $\mathbb{P}$.

\subsection{Smooth entropies}
\label{prel:smoo}
Entropies are used to quantify the uncertainty an observer has about a quantum state. Moreover, conditional entropies quantify the uncertainty of an observer about one subsystem of a bipartite state when he has access to another subsystem. 
The most commonly used quantity is the von Neumann entropy.
Given a state $\rhoAB \in \normstates{\cH_{AB}}$, we denote by $H(A|B)_{\rho} := H(\rhoAB) - H(\rhoB)$ the von Neumann entropy of $A$ conditioned on $B$, where $H(\rho) := - \tr \big( \rho \log \rho \big)$.

While the von Neumann entropy is appropriate for analyzing processes involving a large number of copies of an identical system, the min-entropy is relevant when a single system is considered~\cite{Renner:PHD}.

\begin{definition}[Min-Entropy \cite{Renner:PHD}]\label{conditionalminentropy}
Let $\rhoAB\in\subnormstates{\idx{\cH}{AB}}$, then the min-entropy of $A$ conditioned on $B$ of $\rhoAB$ is defined as
$$\chmin{A}{B}{\rho}\::=\: \max_{\idx{\sigma}{B}\in\normstates{\HB}}\sup\{\lambda\in\mathbb{R}\mid\rhoAB\leq2^{-\lambda}\idA\otimes\sigmaB\}.$$
\end{definition}
More generally, the smooth conditional min-entropy is defined as the largest conditional min-entropy one can get within a distance of at most $\varepsilon$ from $\rho$. Here closeness is measured with respect to the \emph{purified distance}, $P(\rho,\sigma)$, which is defined as~\cite{duality-min-max-entropy}
$$P(\rho,\sigma):= \sqrt{1-\bar{F}(\rho,\sigma)^2},$$
where $\bar{F}(\rho,\sigma)$ is the \emph{generalized fidelity};
$\bar{F}(\rho,\sigma):= F(\rho,\sigma)+\sqrt{(1-\tr\,\rho)(1-\tr\,\sigma)}$ for $\rho,\sigma\in\subnormstates{\cH}$. The purified distance constitutes a 
metric~\cite{duality-min-max-entropy} on $\subnormstates{\cH}$ and satisfies the Fuchs-van de Graaf inequalities
\begin{align}
&\frac{1}{2}\left\|\rho-\sigma\right\|_1+\frac{1}{2}\abs{\tr\,\rho-\tr\,\sigma}\leq P(\rho,\sigma)\nonumber\\
&\leq\sqrt{\left\|\rho-\sigma\right\|_1+\abs{\tr\,\rho-\tr\,\sigma}}\label{lem:genfuchs}.
\end{align}
We say that $\rho$ is $\eps$-close to $\tilde{\rho}$, denoted $\tilde{\rho}\approx_{\eps}\rho$, if $P(\rho,\tilde{\rho})\leq\varepsilon$.
\begin{definition}[Smooth Min-Entropy \cite{Renner:PHD,duality-min-max-entropy}]\label{def:smooth-min-entropy}
Let $\varepsilon\geq0$ and let $\rhoAB\in\subnormstates{\idx{\cH}{AB}}$ with $\sqrt{\tr{\rho}}>\varepsilon$, then the \textit{$\varepsilon$-smooth\ min-entropy} of $A$ conditioned on $B$ of $\rhoAB$ is defined as
$$\chmineps{A}{B}{\rho}\:\:=\:\max_{\tilde{\rho}}\chmin{A}{B}{\tilde{\rho}},$$
where we maximize over all $\tilde{\rho}\approx_{\eps}\rho$.
\end{definition}
Next, we define the smooth max-entropy.
\begin{definition}[Smooth Max-Entropy \cite{Renner:PHD,duality-min-max-entropy}]\label{def:smooth-max-entropy}
Let $\varepsilon\geq0$, let $\rhoAB\in\subnormstates{\idx{\cH}{AB}}$ and let
$\rho_{ABC}\in\subnormstates{\idx{\cH}{ABC}}$ be an arbitrary purification of $\rhoAB$. The \textit{$\varepsilon$-smooth\ max-entropy} of $A$ conditioned on $B$ of $\rhoAB$ is defined as
$$\chmaxeps{A}{B}{\rho}\:\:=\:-\chmineps{A}{C}{\rho}.$$
\end{definition}
The fully quantum asymptotic equipartition property (QAEP) states that in the limit of an infinite number of identical states the smooth min- and max-entropies converge to the von Neumann entropy~\cite{tcr08,marcothesis}: Let $\rhoAB\in\normstates{\cH_{AB}}$, then
\begin{align}
\chh{}{}{A}{B}{\rho}&=\lim_{n\rightarrow\infty}\frac{1}{n} H_{\textnormal{min}}^\varepsilon({A^n}|{B^n})_{\rho^{\otimes n}}\nonumber\\
&=\lim_{n\rightarrow\infty}\frac{1}{n}H_{\textnormal{max}}^\varepsilon({A^n}|{B^n})_{\rho^{\otimes n}}
\label{QAEP}.
\end{align}
In that sense, the smooth conditional min- and max-entropies can be seen as one-shot generalizations of the von Neumann entropy.
\subsection{Uhlmann's theorem and existence of a decoding operation}

To prove a coding theorem it is necessary to establish the existence of a decoding operation. That is, given a quantum state that results from the execution of some quantum channel, we would like to recover the message originally encoded into the input of the channel.
It turns out that this can be achieved if and only if the environment of the channel and some reference system purifying the original message are left uncorrelated after the execution of the channel. In this situation the existence of a decoding operation follows from Uhlmann's Theorem \cite{uhlmann}, which we shall state here for completeness.

\begin{theorem}[Uhlmann's Theorem]
Let $\rho_A,\sigma_A\in\states{\cH_A}$ be two quantum states with respective purifications $\keti{\phi}{\!AB}$ and $\keti{\psi}{\!AC}$. Then,
\begin{align*}
F(\rho_A,\sigma_A)=\max_{V_{B\rightarrow C}}\abs{\bra{\psi}V\ket{\phi}},
\end{align*}
where the maximization goes over all partial isometries from $B$ to $C$.
\end{theorem}
Since our decoupling results involve the Schatten 1-norm rather than the Fidelity it will be useful to transform the above theorem into a statement formulated in terms of Schatten 1-norms. The following Corollary \cite{DHW05} follows from Uhlmann's Theorem with an application of the Fuchs van de Graaf Inequalities (cf. Equation~\eqref{lem:genfuchs}).
\begin{corollary}\label{uhlcor}
Let $\rho_{AB},\:\sigma_{AB}\in\states{\cH_{AB}}$ be pure quantum states and assume that $\Norm{\rhoA-\sigmaA}{1}\leq\varepsilon$. Then there exists some isometry $U_{B\rightarrow C}$ such that $\norm{U \rhoAB\:U^\dagger-\sigma_{AB}}{1}\leq2\sqrt{\varepsilon}$.
\end{corollary}

\section{Dequantizing Theorem}\label{sec:dequantizing-thm}

In this section, we will derive the dequantizing theorem which will be the core technical ingredient for our coding theorems. Our aim will be to derive conditions under which the output of a channel contains only classical correlations with a reference system. More precisely, we will prove the following:

\begin{theorem}[Dequantizing Theorem]\label{thm:dequantizing}
	Let $\cT_{A\rightarrow EX}$ be a complementary CQ channel, and let $\omega_{A'EX}\in\subnormstates{\cH_{A'EX}}$ be its Choi-Jamiołkowski representation. Let $\rho_{AR}$ be a pure state on $\cH_{AR}$, and let $\rho_{AR}^{\cl} := \cC_A(\rho_{AR})$. Then
\begin{multline}
	\mbE_{P_A \in \mbP}{\Norm{\mathcal{T} \left(\idx{P}{A} (\idx{\rho}{AR} - \idx{\rho}{AR}^{\cl})\ \idx{P}{A}^\dagger \right)}{1}}\\
\leq\sqrt{\frac{1}{\idx{d}{A}-1}\:2^{-\chmin{A'}{EX}{\omega}-\chmin{A}{R}{\rho}}},
\end{multline}
where the permutation operators act by permuting the Schmidt-basis vectors of $\rhoAR$.
\end{theorem}

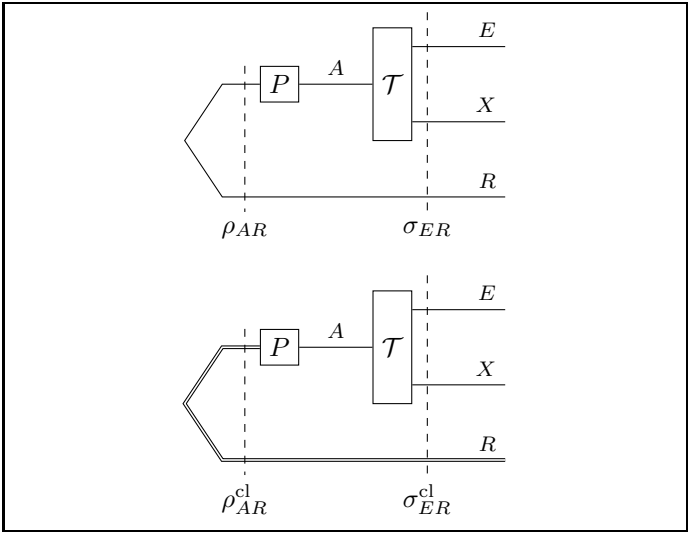
\begin{figure}
	\centering

		\begin{tikzpicture}[auto]
		\tikzstyle{gate} = [fill=white, draw]
		\tikzstyle{syslabel} = [font=\footnotesize]

		\coordinate (rightwall) at (5,0);

		\node[gate] (encoder) at (2,3) {$P$};
		\node[gate,minimum height=1.5 cm] (channel) at (3.5,3) {$\mathcal{T}$};

		\draw (encoder.west) -- ++(-.5,0) -- ++(-.5,-.75) -- ++(.5,-.75) coordinate (bottomline) to (rightwall |- bottomline) node[syslabel,above left] {$R$};
		\draw (encoder) to node[syslabel] {$A$} (channel);
		\draw ([yshift=.5cm] channel.east) to ([yshift=.5cm] rightwall |- channel.east) node[syslabel,above left] {$E$};
		\draw ([yshift=-.5cm] channel.east) to ([yshift=-.5cm] rightwall |- channel.east) node[syslabel,above left] {$X$};


		\draw[dashed] ([xshift=.2cm,yshift=.2cm] channel.east |- channel.north) -- ([xshift=.2cm,yshift=-.2cm] channel.east |- bottomline) node[below] {$\sigma_{ER}$};

		\draw[dashed] ([xshift=-.2cm] encoder.west |- encoder.north) -- ([xshift=-.2cm,yshift=-.2cm] encoder.west |- bottomline) node[below] {$\rho_{AR}$};

		\begin{scope}[yshift=-3.5cm]
		\coordinate (rightwall) at (5,0);

		\node[gate] (encoder) at (2,3) {$P$};
		\node[gate,minimum height=1.5cm] (channel) at (3.5,3) {$\mathcal{T}$};

		\draw[double] (encoder.west) -- ++(-.5,0) -- ++(-.5,-.75) -- ++(.5,-.75) coordinate (bottomline) to (rightwall |- bottomline) node[syslabel,above left] {$R$};
		\draw (encoder) to node[syslabel] {$A$} (channel);
		\draw ([yshift=.5cm] channel.east) to ([yshift=.5cm] rightwall |- channel.east) node[syslabel,above left] {$E$};
		\draw ([yshift=-.5cm] channel.east) to ([yshift=-.5cm] rightwall |- channel.east) node[syslabel,above left] {$X$};


		\draw[dashed] ([xshift=.2cm,yshift=.2cm] channel.east |- channel.north) -- ([xshift=.2cm,yshift=-.2cm] channel.east |- bottomline) node[below] {$\sigma^{\cl}_{ER}$};

		\draw[dashed] ([xshift=-.2cm] encoder.west |- encoder.north) -- ([xshift=-.2cm,yshift=-.2cm] encoder.west |- bottomline) node[below] {$\rho^{\cl}_{AR}$};

		\end{scope}

		\end{tikzpicture}

		\caption{Illustration of the dequantizing theorem. The top diagram illustrates the situation in which we apply the dequantizing theorem: we apply a random permutation $P$ to $\rho_{AR}$ following by the channel. The bottom diagram illustrates the ``ideal'' state we would like to get at the end: a state containing only classical correlations between $R$ and $E$.}
\label{fig:dequantizing-comparison}
\end{figure}


In other words, Theorem \ref{thm:dequantizing} gives a bound on how close the state $\cT(P_A \rho_{AR} P_A^{\dagger})$ is from a state containing only classical correlations between $R$ and $E$ (namely, $\cT(P_A \rho_{AR}^{\cl} P_A^{\dagger})$). See Figure \ref{fig:dequantizing-comparison} for an illustration.

The rest of this section is devoted to the proof of Theorem \ref{thm:dequantizing} and is organized in three subsections. In the first one we calculate the above expectation value with the Schatten 1-norm replaced by the Schatten 2-norm. We conclude the proof of the above theorem in the second subsection showing how the statement found about the Schatten 2-norm can be transformed into the one above. In the third subsection we reformulate the upper bound of Theorem~\ref{thm:dequantizing} using the smooth conditional min-entropy. This enables us to make statements about independent, identically distributed channels via the QAEP, Equation~\eqref{QAEP}.

\subsection{Dequantizing with Schatten 2-Norms}

We first prove a statement that holds for general hermiticity preserving, linear maps $\cN_{A\rightarrow B}\in\textnormal{Hom}(\linops{\HA},\linops{\HB})$ with Choi-Jamiołkowski representation $\omega_{A'B}\in\hermops{\cH_{A'B}}$. In Proposition \ref{lemdist} below, we will compute the expectation value
\begin{multline*}
\EP{\Norm{\mathcal{N}\left(\idx{P}{A} (\idx{\rho}{AR} - \idx{\rho}{AR}^{\cl}) \idx{P}{A}^\dagger\right)}{2}^2}\\
=\frac{\idx{d}{A}}{\idx{d}{A}-1}\:{\Norm{\rho_{AR}-\rho_{AR}^{\cl}}{2}^2}\:\Norm{\idx{\omega}{A'B}-\idx{\omega}{A'B}^{\cl}}{2}^2,
\end{multline*}
where the permutation operators act by permuting the basis vectors of the Schmidt-basis\footnote{An extension of this result to an arbitrary basis is known \cite{mythesis}.}. To prove this, we first need the following lemma:

\begin{lemma}
 \label{counting}Let $\idx{\cH}{A}$ be a Hilbert space with orthonormal basis $\{\ket{i}\}_{i=1,...,\idx{d}{A}}$ and let $\mathbb{P}$ be the corresponding set of permutation operators. Then for any $i\neq j$ one has that
\begin{multline*}
	\mbE_{P \in \mbP}\left(P^{\otimes2}\:(\proji{i}{j}{A}\otimes\proji{j}{i}{A'})\:(P^\dagger)^{\otimes2}\right) \\
=\frac{1}{\idx{d}{A}(\idx{d}{A}-1)}\left(F_{AA'}-d_A\:T_{AA'}\right).
\end{multline*}
\end{lemma}

\begin{IEEEproof}
There are $\idx{d}{A}!$ permutation operators in $\mathbb{P}$. For arbitrary but fixed $i\neq j$ and $k\neq l$ there are $(\idx{d}{A}-2)!$ permutation operators with
$(\idx{P}{A})^{\otimes2}\:\ket{i}_{A}\otimes\ket{j}_{A'}=\ket{k}_A\otimes\ket{l}_{A'}$.
On the other hand there is no permutation such that for $i\neq j$ the operator $(\idx{P}{A})^{\otimes2}$ maps $\ket{i}_{A}\otimes\ket{j}_{A'}$ to $\ket{k}_A\otimes\ket{l}_{A'}$ with $k=l$. We conclude that
\begin{align*}
&\EP\left((\idx{P}{A})^{\otimes2}\:(\proji{i}{j}{A}\otimes\proji{j}{i}{A'})\:(\idx{P}{A}^\dagger)^{\otimes2}\right)\\
&=\frac{(\idx{d}{A}-2)!}{\idx{d}{A}!}\sum_{k\neq l}^{\idx{d}{A}}\proji{k}{l}{A}\otimes\proji{l}{k}{A'}.
\end{align*}
\end{IEEEproof}

\begin{proposition}[Distance from classicality]\label{lemdist}
	Let $\cN_{A\rightarrow B}\in\textnormal{Hom}(\linops{\HA},\linops{\HB})$ be a linear map with Choi-Jamiołkowski representation $\omega_{A'B}\in\hermops{\cH_{BA'}}$, and let $\ket{\rho}_{AR} = \sum_i \sqrt{\lambda_i} \ket{ii}_{AR}$. Then
\begin{multline*}
\EP{\Norm{\mathcal{N}\left(\idx{P}{A} (\idx{\rho}{AR} - \idx{\rho}{AR}^{\cl}) \idx{P}{A}^\dagger \right)}{2}^2}\\
=\frac{\idx{d}{A}}{\idx{d}{A}-1}\:{\Norm{\rho_{AR}-\rho_{AR}^{\cl}}{2}^2}\:\Norm{\idx{\omega}{A'B}-\idx{\omega}{A'B}^{\cl}}{2}^2,
\end{multline*}
where the operators $P$ permute the Schmidt-basis vectors of $\rhoA$.
\end{proposition}

\begin{IEEEproof}
Rewriting the Schatten 2-norm in terms of the trace, we get 

\begin{align}
&\EP{\Norm{\mathcal{N}(\idx{P}{A} \otimes \idi{R} \ (\idx{\rho}{AR}-\idx{\rho}{AR}^{\cl})\ \idx{P}{A}^\dagger \otimes \idi{R})}{2}^2}\nonumber\\
&=\EP{\Trace{\mathcal{N}(\idx{P}{A} \otimes \idi{R} \ (\idx{\rho}{AR}-\idx{\rho}{AR}^{\cl})\ \idx{P}{A}^\dagger \otimes \idi{R})^2}}\nonumber\\
&=\EP{\Trace{\mathcal{N}\left(\sum_{i\neq j}^{\idx{d}{R}}{\sqrt{\lambda_i\lambda_j}\:\idx{P}{A}\proji{i}{j}{A}\idx{P}{A}^\dagger \otimes\proji{i}{j}{R}}\right)^2}}\nonumber\\
&=\sum_{i\neq j}^{\idx{d}{R}}\lambda_i\lambda_j\left[\EP{\Trace{\mathcal{N}\left(\idx{P}{A}\proji{i}{j}{A}\idx{P}{A}^\dagger\right)\mathcal{N}\left(\idx{P}{A}\proji{j}{i}{A}\idx{P}{A}^\dagger\right)}}\right]\nonumber\\
&=\sum_{i\neq j}^{\idx{d}{R}}\Big\{\lambda_i\lambda_j\ \cdot\nonumber\\
&\quad\Trace{\mathcal{N}^{\otimes2}\left(\EP\left(\idx{P}{A}^{\otimes2}(\proji{i}{j}{A}\otimes\proji{j}{i}{A'})(\idx{P}{A}^\dagger)^{\otimes2}\right)\right)\idx{F}{BB\rq{}}}\Big\}\label{swaptrick}\\
&=\sum_{i\neq j}^{\idx{d}{R}}\frac{\lambda_i\lambda_j}{\idx{d}{A}(\idx{d}{A}-1)}\Trace{\mathcal{N}^{\otimes2}\left(\left(F_{AA\rq{}}-d_A\:T_{AA'}\right)\right)\idx{F}{BB\rq{}}}\label{expectvalue}\\
&=\frac{\norm{\rho_{AR}-\rho_{AR}^{\cl}}{2}^2}{d_A(\idx{d}{A}-1)}\ \cdot\nonumber\\
&\quad\left(\Trace{\mathcal{N}^{\otimes2}\left(\idx{F}{AA\rq{}}\right)\idx{F}{BB\rq{}}}-\Trace{(\mathcal{N}^{\cl})^{\otimes2}\left(\idx{F}{AA\rq{}}\right)\idx{F}{BB\rq{}}}\right)\label{useCJINV}.
\end{align}

Equation \eqref{swaptrick} is by an application of the swap trick and in equation \eqref{expectvalue} we applied Lemma~\ref{counting}. To simplify \eqref{useCJINV} we evaluate the term $\Trace{\mathcal{N}^{\otimes2}\left(\idx{F}{AA\rq{}}\right)\idx{F}{BB\rq{}}}$ using the inverse Choi-Jamiołkowski isomorphism:
\begin{align}
&\Trace{\mathcal{N}^{\otimes2}\left(\idx{F}{AA\rq{}}\right)\idx{F}{BB\rq{}}}\nonumber\\
&=\idx{d}{A}^2\:\Trace{\Ptrace{AA'}{\idx{\omega}{AB}^{\otimes2}\left(\idx{F}{AA\rq{}}\otimes\idi{BB'}\right)}\idx{F}{BB\rq{}}}\nonumber\\
&=\idx{d}{A}^2 \Trace{\idx{\omega}{AB}^{\otimes2}\left(\idx{F}{AA\rq{}}\otimes\idi{BB'}\right)\left(\idi{AA'}\otimes\idx{F}{BB\rq{}}\right)}\label{ptradjoint}\\
&=\idx{d}{A}^2\:\Trace{\idx{\omega}{AB}^{2}}\nonumber.
\end{align}
In Equation \eqref{ptradjoint} we used the fact that the adjoint mapping of the partial trace is tensoring with the identity.
Analogously, one has
\begin{align*}
\Trace{(\mathcal{N}^{\cl})^{\otimes2}\left(\idx{F}{AA\rq{}}\right)\idx{F}{BB\rq{}}}=\idx{d}{A}^2\:\Trace{(\idx{\omega}{AB}^{\cl})^{2}}
\end{align*}
and the second factor of \eqref{useCJINV} becomes
\begin{align*}
&\Trace{\mathcal{N}^{\otimes2}\left(\idx{F}{AA\rq{}}\right)\idx{F}{BB\rq{}}}-\Trace{(\mathcal{N}^{\cl})^{\otimes2}\left(\idx{F}{AA\rq{}}\right)\idx{F}{BB\rq{}}}\\
&=\idx{d}{A}^2\left(\Trace{\idx{\omega}{AB}^2}-\Trace{(\idx{\omega}{AB}^{\cl})^2}\right)\\
&=\idx{d}{A}^2\Norm{\idx{\omega}{AB}\:-\:\idx{\omega}{AB}^{\cl}}{2}^2.
\end{align*}
\end{IEEEproof}

\subsection{Dequantizing with the Schatten 1-norm}
In this subsection we derive Theorem~\ref{thm:dequantizing} with an application of Proposition~\ref{lemdist}. We use the following lemma.

\begin{lemma}\label{minentr}
For any $\idx{\xi}{AR}\in\subnormstates{\idx{\cH}{AR}}$, there exists an operator $\idx{\zeta}{R}\in\normstates{\idx{\cH}{R}}$ with
$$\frac{1}{\tr{[\idx{\xi}{AR}]}}\Trace{\left((\idi{A}\otimes\zetaR^{-1/2})\idx{\xi}{AR}\right)^2}\leq2^{-\chmin{A}{R}{\xi}}.$$
\end{lemma}

\begin{IEEEproof}
Choose $\zetaR$ such that it maximizes the min-entropy, i.e.\ it satisfies
$\xi_{AR} \leq2^{-\chmin{A}{R}{\xi}}\idi{A}\otimes\zetaR $.
Hence,
\begin{align*}
\sqrt{\xi_{AR}}(\idi{A}\otimes\zetaR^{-\frac{1}{2}})\xi_{AR}(\idi{A}\otimes\zetaR^{-\frac{1}{2}})\sqrt{\xi_{AR}}\leq2^{-\chmin{A}{R}{\xi}}\xi_{AR}
\end{align*}
Taking the trace on both sides concludes the proof.
\end{IEEEproof}

\begin{IEEEproof}[Proof of Theorem~\ref{thm:dequantizing}]
We first introduce some notation. We abbreviate the difference between $\rho_{AR}$ and its classicalized version by writing $\bar{\rho}_{AR}:=\rho_{AR}-\rho_{AR}^{\cl}$. By Lemma~\ref{minentr} there are operators $\sigma_{EX}$ and $\tau_R$ such that 
$$\frac{1}{\tr{[\omega_{A'EX}]}}\Trace{\left((\idi{A'}\otimes\sigma_{EX}^{-1/2})\omega_{A'EX} \right)^2}\leq2^{-\chmin{A'}{EX}{\omega}}$$
and
$$\frac{1}{\tr{[\rho_{AR}]}}\Trace{\left((\idi{A}\otimes\tau_{R}^{-1/2})\rhoAR \right)^2}\leq2^{-\chmin{A}{R}{\rho}}.$$
Since $\bar{\cT}_{A'\rightarrow EX}$ is a complementary channel of a CQ-channel we can assume that $\sigma_{EX}$ has CQ-structure, i.e. $\sigma_{EX}=\sum_x\sigma_E^x\otimes\proji{x}{x}{X}$. Furthermore the operator $\rhoAR$ is given in its Schmidt-basis, such that $\tau_R$ can be written as $\tau_R=\sum_xr_x\proji{x}{x}{R}$. (Both facts follow from Lemma~\ref{marcoslemma} in Appendix~\ref{A} with $\varepsilon=0$.)

We introduce a system $P$ with $(\idx{d}{A}!)$-dimensional Hilbert space, $\idx{\cH}{P}$, and canonical basis vectors $\ket{P}_P$ that correspond to the permutation operators of $P\in\mathbb{P}$. We define the operator
\begin{multline*}
\idx{\zeta}{PREX}:=\\
\EP{\left(\sum_{x=1}^{\idx{d}{R}}r_x\proj{P}{P}_{P}\otimes\proj{x}{x}_{R}\otimes\idx{\sigma}{E}^{P(x)}\otimes\proj{P(x)}{P(x)}_{X}\right)} \,,
\end{multline*}
which can be inverted on its support to yield
\begin{multline*}
\idx{\zeta}{PREX}^{-1}:=d_{A!}\sum_{P\in\mathbb{P}}\Big(\sum_{x=1}^{\idx{d}{R}}r_x^{-1}\proj{P}{P}_{P}\otimes\proj{x}{x}_{R}\ \otimes\\
\qquad\qquad\left(\idx{\sigma}{E}^{P(x)}\right)^{-1}\otimes\proj{P(x)}{P(x)}_{X}\Big) \,.
\end{multline*}

We note that the operator $\zeta^{\frac{1}{4}}\zeta^{-\frac{1}{4}}$ is a projector and 
\begin{align}
&\EP{\proji{P}{P}{P}\otimes\bar{\mathcal{T}}(\idx{P}{A} \otimes \idi{R}\bar{\rho}_{AR}\idx{P}{A}^\dagger \otimes \idi{R})}\nonumber\\
&=(\zeta^{\frac{1}{4}}\zeta^{-\frac{1}{4}})\left[\EP{\proji{P}{P}{P}\otimes\bar{\mathcal{T}}(\idx{P}{A} \otimes \idi{R}\bar{\rho}_{AR}\idx{P}{A}^\dagger \otimes \idi{R})}\right](\zeta^{-\frac{1}{4}}\zeta^{\frac{1}{4}})\label{gotridofPi}
\end{align}

Using these operators, we write
\begin{multline}
\EP{\Norm{\bar{\mathcal{T}}(\idx{P}{A} \bar{\rho}_{AR} \idx{P}{A}^\dagger)}{1}}\\
\begin{split}
&=\EP{\Norm{\proji{P}{P}{P}\otimes\left(\bar{\mathcal{T}}(\idx{P}{A} \bar{\rho}_{AR} \idx{P}{A}^\dagger)\right)}{1}}\\
&=\Norm{\EP{\proji{P}{P}{P}\otimes\left(\bar{\mathcal{T}}(\idx{P}{A} \bar{\rho}_{AR} \idx{P}{A}^\dagger)\right)}}{1}\\
&\leq\sqrt{\Trace{\idx{\zeta}{PREX}}} \Norm{{\zeta}^{-\frac{1}{4}}\left(\EP\proji{P}{P}{P}\otimes\bar{\mathcal{T}}(\idx{P}{A} \bar{\rho}_{AR} \idx{P}{A}^\dagger)\right){\zeta}^{-\frac{1}{4}}}{2}\label{plugback}.
\end{split}
\end{multline}

Inequality~\eqref{plugback} follows from Equation~\eqref{gotridofPi} together with an application of the H\"older-type inequality $\Norm{ABC}{1}\leq\Norm{\abs{A}^4}{1}^{\frac{1}{4}}\Norm{\abs{B}^2}{1}^{\frac{1}{2}}\Norm{\abs{C}^4}{1}^{\frac{1}{4}}$, \cite{bhatia}.
%
%
%
%
%
%
The trace term on the right hand side of Inequality~\eqref{plugback} can be evaluated directly to be
\begin{align*}
\Trace{\idx{\zeta}{PREX}}&=\sum_{x=1}^{d_R} r_x\, \EP{\Trace{\idx{\sigma}{E}^{P(x)}}}=\frac{1}{d_A}.
\end{align*}
%
Thus, it is sufficient to evaluate the term with the Schatten 2-norm. For notational convenience we introduce the map $\tilde{\cT}(\cdot):=(\sigma_{EX})^{-\frac{1}{4}}\bar{\cT}(\cdot)(\sigma_{EX})^{-\frac{1}{4}}$ with Choi-Jamiołkowski representation $\tilde{\omega}_{A'EX}$ and the operator $\tilde{\rho}_{AR}:=(\idi{A}\otimes\tau_R)^{-\frac{1}{4}}\bar{\rho}_{AR}(\idi{A}\otimes\tau_R)^{-\frac{1}{4}}$.
Using the fact that $\bar{\cT}$ is the complementary channel of a CQ-channel one can verify that
\begin{multline*}
\zeta^{-\frac{1}{4}}\left(\EP\proji{P}{P}{P}\otimes\bar{\mathcal{T}}(\idx{P}{A}\,\bar{\rho}_{AR}\,\idx{P}{A}^\dagger)\right)\zeta^{-\frac{1}{4}}\\
\begin{split}
&=\frac{1}{\sqrt{d_A!}}\sum_{P\in\mathbb{P}}\proji{P}{P}{P}\, \otimes\\
&\quad\left((\sigma_{EX}\!\otimes\!\tau_{R})^{-\frac{1}{4}}\bar{\mathcal{T}}( \idx{P}{A}\, \bar{\rho}_{AR} \idx{P}{A}^\dagger)(\sigma_{EX}\!\otimes\!\tau_{R})^{-\frac{1}{4}}\right)\\
&=\frac{1}{\sqrt{d_A!}}\sum_{P\in\mathbb{P}}\proji{P}{P}{P}\otimes\mathcal{\widetilde{T}}(\idx{P}{A} \,\tilde{\rho}_{AR}\,\idx{P}{A}^\dagger) .
\end{split}
\end{multline*}
Using this, we find
\begin{multline*}
\Norm{\zeta^{-\frac{1}{4}}\left(\EP\proji{P}{P}{P}\otimes\bar{\mathcal{T}}(\idx{P}{A} \:\bar{\rho}_{AR}\:\idx{P}{A}^\dagger )\right)\zeta^{-\frac{1}{4}}}{2}^{2}\\
\begin{split}
&=\EP{\Norm{\mathcal{\widetilde{T}}(\idx{P}{A}  \:\tilde{\rho}_{AR}\:\idx{P}{A}^\dagger )}{2}^2}\\
&=\frac{\idx{d}{A}}{\idx{d}{A}-1}\Norm{\idx{\tilde{\omega}}{A'EX}\:-\:\idx{\tilde{\omega}}{A'EX}^{\cl}}{2}^2\Norm{\idx{\tilde{\rho}}{AR}\:-\:\idx{\tilde{\rho}}{AR}^{\cl}}{2}^2\\
&\leq\frac{\idx{d}{A}}{\idx{d}{A}-1}\:\Trace{\idx{\tilde{\omega}}{A'EX}^2}\Trace{\idx{\tilde{\rho}}{AR}^2}\\
&\leq\frac{\idx{d}{A}}{\idx{d}{A}-1}\:2^{-\chmin{A'}{EX}{\omega}-\chmin{A}{R}{\rho}},
\end{split}
\end{multline*}
where we apply Proposition~\ref{lemdist} to obtain the second equality and use the special choice of $\sigma_{EX}$ and $\tau_R$ (cf. Lemma~\ref{minentr}) for the last inequality. Plugging this into Equation~\eqref{plugback} concludes the proof.
\end{IEEEproof}

\subsection{A smoothed version of the dequantizing theorem}
The reason for introducing smooth versions of the min- and max-entropy is that these quantities can vary a lot with small variations in the underlying states while the quantities that we are bounding with them do not. This is such a case; it is therefore desirable to have a version of the dequantizing theorem which involves the smooth entropies. The smooth quantities have the additional advantage that they converge to the corresponding von Neumann quantities in the i.i.d. case (cf. Equation~\ref{QAEP}). We therefore prove the following:

\begin{theorem}\label{thm:maintheorem}
	Let $\bar{\cT}_{A\rightarrow EX}$ be a complementary channel of a CQ-channel, let $\omega_{A'EX}\in\subnormstates{\cH_{EXA'}}$ be the Choi-Jamiołkowski representation of $\bar{\cT}$, and let $\ket{\rho}_{AR} = \sum_i \sqrt{\lambda_i} \ket{ii}_{AR}$. Let $\eps, \eps'$ be such that
$\sqrt{\trace{\rho}} > \eps \geq 0$ and $\sqrt{\trace{\omega}}\ >\eps' \geq0$. Then, 
\begin{multline*}
\EP{\Norm{\bar{\mathcal{T}}\left(\idx{P}{A} (\idx{\rho}{AR} - \idx{\rho}{AR}^{\cl}) \idx{P}{A}^\dagger \right)}{1}}\\
\leq\sqrt{\frac{1}{\idx{d}{A}-1}\:2^{-\chmineeps{\eps'}{A'}{EX}{\omega}-\chmineps{A}{R}{\rho}}}+8\eps + 8\eps',
\end{multline*}
where the permutation operators act by permuting the Schmidt-basis vectors of $\rhoA$.
\end{theorem}

\begin{IEEEproof}
Let $\idx{\widehat{\omega}}{A'EX}\in\subnormstates{\idx{\cH}{A'EX}}$ be a state that saturates the bound in the definition of the smooth min-entropy, i.e.\:$P(\idx{\omega}{A'EX},\idx{\widehat{\omega}}{A'EX})\leq\eps'$ and $\chmin{A'}{EX}{\widehat{\omega}} = \chmineeps{\eps'}{A'}{EX}{\omega}$. Analogously, $\idx{\widehat{\rho}}{AR}$ satisfies $P(\idx{\widehat{\rho}}{AR},\idx{\rho}{AR})\leq\eps$ and $\chmin{A}{R}{\widehat{\rho}} = \chmineeps{\eps}{A}{R}{\rho}$.

Using inequality \eqref{lem:genfuchs}, we find that
\begin{align}
\left\|\idx{\omega}{A'EX}-\idx{\widehat{\omega}}{A'EX}\right\|_1 &\leq 2\varepsilon' & \left\|\idx{\rho}{AR}-\idx{\widehat{\rho}}{AR}\right\|_1 &\leq 2\varepsilon\label{omrobnd}.
\end{align}
We decompose $\widehat{\omega}-\omega$ into positive operators with orthogonal support writing $\widehat{\omega}-\omega = \Delta_+-\Delta_-$ and conclude from \eqref{omrobnd} that $\Norm{\Delta_+}{1}\leq2\varepsilon'$ and $\Norm{\Delta_-}{1}\leq2\varepsilon'$.

Similarly, we decompose $(\rho_{AR}-\rho_{AR}^{\cl})-(\widehat{\rho}_{AR}-\widehat{\rho}_{AR}^{\cl})=\Gamma_+-\Gamma_-$ with the operators $\Gamma_+$ and $\Gamma_-$ again chosen to be positive and with orthogonal support. From the second inequality in \eqref{omrobnd}, we conclude that
$\Norm{\Gamma_+}{1} \leq4\varepsilon$ and $\Norm{\Gamma_-}{1} \leq4\varepsilon
$.

Let $\widehat{\cT}$, $\cD_+$ and $\cD_-$ be the unique Choi-Jamiołkowski preimages of $\idx{\widehat{\omega}}{A'EX}$, $\Delta_+$ and $\Delta_-$ respectively. 
We note that from the fact that $\omega_{A\rq{}EX}$ is classically coherent between $A\rq{}$ and $X$ it follows that $\widehat{\omega}_{A\rq{}EX}$ also shares this property (See Appendix~\ref{A}, Lemma~\ref{marcoslemma}). Furthermore the state $\rho_{AR}$ is classically coherent between $A$ and $R$, such that the state $\widehat{\rho}_{AR}$  has the same Schmidt-basis as $\rho_{AR}$ (Lemma~\ref{marcoslemma}). We therefore can apply Theorem~\ref{thm:dequantizing} on the states $\widehat{\omega}$ and $\widehat{\rho}$ to find
\begin{multline*}
\sqrt{\frac{1}{\idx{d}{A}-1}\:2^{-\chmineeps{\eps'}{A'}{EX}{\omega}-\chmineps{A}{R}{\rho}}}\\
\begin{split}
&=\sqrt{\frac{1}{\idx{d}{A}-1}\:2^{-\chmin{A'}{EX}{\widehat{\omega}}-\chmin{A}{R}{\widehat{\rho}}}}\\
&\geq\EP{\Norm{\widehat{\mathcal{T}}\left(\idx{P}{A} (\widehat{\rho}_{AR}-\widehat{\rho}_{AR}^{\cl})\:\idx{P}{A}^\dagger \right)}{1}}.
\end{split}
\end{multline*}
Applying the triangle inequality twice shows that for any permutation operator we have
\begin{align}
&\Norm{\widehat{\mathcal{T}}\left(\idx{P}{A} (\widehat{\rho}_{AR}-\widehat{\rho}_{AR}^{\cl}) \idx{P}{A}^\dagger \right)}{1}\nonumber\\
&\geq\Norm{\bar{\mathcal{T}}\left(\idx{P}{A} (\rhoAR-\rhoAR^{\cl}) \idx{P}{A}^\dagger \right)}{1}\nonumber\\
&\quad-\Norm{\left(\bar{\mathcal{T}}-\widehat{\mathcal{T}}\right)\left(\idx{P}{A} (\widehat{\rho}_{AR}-\widehat{\rho}_{AR}^{\cl}) \idx{P}{A}^\dagger \right)}{1}\nonumber\\
&\quad-\Norm{\bar{\mathcal{T}}\left(\idx{P}{A} \left((\rho_{AR}-\rho_{AR}^{\cl})-(\widehat{\rho}_{AR}-\widehat{\rho}_{AR}^{\cl})\right)\idx{P}{A}^\dagger \right)}{1}\label{findthrill}
\end{align}

The first term on the right hand-side of Inequality~\eqref{findthrill} corresponds to the unsmoothed dequantizing theorem. For the remaining two terms 
we find upper bounds:
\begin{multline}
\EP{\Norm{\left(\bar{\mathcal{T}}-\widehat{\mathcal{T}}\right)\left(\idx{P}{A} (\widehat{\rho}_{AR}-\widehat{\rho}_{AR}^{\cl}) \idx{P}{A}^\dagger \right)}{1}}\\
\begin{split}
&\leq\sum_{a\in\{+,-\}}{\EP{\Trace{D_a\left(\idx{P}{A} \widehat{\rho}_{AR} \idx{P}{A}^\dagger \right)}}}\\
&\quad+\sum_{a\in\{+,-\}}{\EP{\Trace{D_a\left(\idx{P}{A} \widehat{\rho}_{AR}^{\cl} \idx{P}{A}^\dagger \right)}}}\\
&=2\sum_{a\in\{+,-\}}{\EP{\Trace{D_a\left(\idx{P}{A}\:\widehat{\rho}_{A}\:\idx{P}{A}^\dagger\right)}}}\\
&\leq 2\:\left(\trace{\Delta_+}+\trace{\Delta_-}\right)\leq8\varepsilon'\label{frstsm}
\end{split}
\end{multline}

We bound the third term in a similar way. We have
\begin{multline}
\EP{\Norm{\bar{\mathcal{T}}\left(\idx{P}{A} \left((\rho_{AR}-\rho_{AR}^{\cl})-(\widehat{\rho}_{AR}-\widehat{\rho}_{AR}^{\cl})\right)\idx{P}{A}^\dagger \right)}{1}}\\
\begin{split}
&\leq\sum_{a\in\{+,-\}}{\EP{\Trace{\bar{\mathcal{T}}\left(\idx{P}{A} \Gamma_a \idx{P}{A}^\dagger \right)}}}\\
&\leqslant \trace{\Gamma_+}+\trace{\Gamma_-}\leq8\varepsilon\label{scndsm}.
\end{split}
\end{multline}

Substituting the expressions \eqref{frstsm} and \eqref{scndsm} into Inequality~\eqref{findthrill} shows that
\begin{multline*}
\EP{\Norm{\widehat{\mathcal{T}}\left(\idx{P}{A} (\widehat{\rho}_{AR}-\widehat{\rho}_{AR}^{\cl})\:\idx{P}{A}^\dagger \right)}{1}}\\
\geq\EP{\Norm{\bar{\mathcal{T}}\left(\idx{P}{A} (\rho_{AR}-\rho_{AR}^{\cl})\:\idx{P}{A}^\dagger \right)}{1}}-8\varepsilon-8\eps'
\end{multline*}
and an application of Theorem~\ref{thm:dequantizing} concludes the proof.
\end{IEEEproof}

\section{From dequantizing to coding}
\label{sec:proof-hsw}

In the following subsections we show how the dequantizing theorem (Theorem \ref{thm:maintheorem}) yields a one-shot coding theorem for classical information. Then, in the last subsection, we apply this coding theorem to the iid scenario and obtain the HSW theorem as a corollary.

\subsection{Sending classical information through a quantum channel}\label{HSW-scenario}
Consider the following scenario: Alice wants to send a classical message $M$ to Bob using a quantum channel $\idx{\cT}{A\rightarrow B}$. For this purpose she encodes her message using an encoding TPCPM $\idx{\cE}{M\rightarrow A}$ into the input $A$ of the channel. Having received the output of the channel $B$, Bob will apply a decoding TPCPM $\idx{\cD}{B\rightarrow\widehat{M}}$ aiming to recover the original message. Since we are interested in the transmission of classical data through a quantum channel, the operation $\cE$ can be assumed to be classical, while $\cT$ is a CQ-channel.

Alice\rq{}s message $M$ is assumed to exhibit perfect correlations with some classical reference system $R$. This means that the joint state of the message and the reference can be represented by the operator $\varphi_{MR}^{\cl}=\sum_{i=1}^{d_M}\lambda_i \proji{ii}{ii}{MR}$ for some probability distribution $\lambda$. 
The aim is that after decoding, Bob holds a system $\widehat{M}$ which contains Bob's decoded message. Naturally, we want $\widehat{M}$ to contain the same message as $M$ with high probability, which is equivalent to saying that $\widehat{M}$ is almost perfectly correlated to the reference system $R$. (See Figure \ref{fig:classical-problem} for an illustration of this scenario.) Mathematically, this means that we want the probability of error $p_e$ to be bounded as
\begin{align*}
	2p_e = {\left\|(\idx{\cD}{B\rightarrow \widehat{M}}\circ\idx{\cT}{A\rightarrow B}\circ\idx{\cE}{M\rightarrow A})(\varphi^{\cl}_{MR})\ -\ \varphi_{\widehat{M}R}^{\cl} \right\|}_1\ \leq\ \varepsilon\label{closea}.
\end{align*}
In other words, the state after encoding, the channel, and decoding is within $\varepsilon$ in trace distance to a state that is perfectly correlated between $\widehat{M}$ and the copy of the message in $R$.

\begin{figure}
	\centering
		\begin{tikzpicture}[auto]
		\tikzstyle{gate} = [fill=white, draw]
		\tikzstyle{syslabel} = [font=\footnotesize]

		\node[gate] (encoder) at (2,3) {$\mathcal{E}$};
		\node[gate] (channel) at (3.5,3) {$\mathcal{T}$};
		\node[gate] (decoder) at (5,3) {$\mathcal{D}$};

		\draw[double] (encoder) -- (1,3) -- (0.5,2) -- (1,1) -- (6,1);
		\draw[double] (encoder) to node[syslabel] {$A$} (channel);
		\draw (channel) to node[syslabel] {$E$} (decoder);
		\draw[double] (decoder) -- (6,3) node[syslabel,above left] {$\widehat{M}$};

		\node[syslabel,above right] at (1,3) {$M$};
		\node[syslabel,above right] at (1,1) {$R$};
		\node[left] at (0.5,2) {$\varphi^{\cl}_{MR}$};


		\end{tikzpicture}
		\caption{Diagram illustrating the transmission of classical data through a quantum channel. The state $\varphi^{\cl}_{MR}$ represents perfect classical correlations of a message $M$ and a reference $R$. The aim is to obtain after decoding a system $\widehat{M}$ with nearly perfect classical correlations to $R$.}
	\label{fig:classical-problem}
\end{figure}
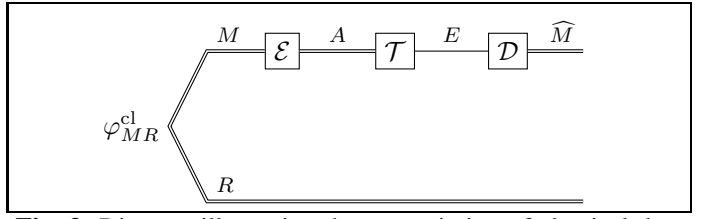



\subsection{The purified picture}
To apply the dequantizing theorem, it is necessary to work with pure states and operations. Hence, for our derivation we will consider the setup depicted in Figure \ref{fig:classical-problem}, but where all states and operations are replaced with the corresponding purifications.

The present subsection shows how the dequantizing theorem (Theorem~\ref{thm:dequantizing}) can be used to show the existence of a decoding operation; in the following subsection, we will apply the theorem to get the encoder.

First, we purify the state $\varphi_{MR}^{\cl}$ from above to $\ket{\varphi}_{MR}=\sum_{i}\sqrt{\lambda_i}\ket{ii}_{MR}$. Next, we will assume that the encoder $\mathcal{E}$ is actually a partial isometry $V_{M \rightarrow A}$; this is slightly less general, but it will turn out to be enough for our purposes. Then, we replace $\mathcal{T}$ by its Stinespring dilation $U^\cT_{A \rightarrow BXE}$, where $X$ contains a copy of the classical input, as explained in the preliminaries (Section \ref{prel}). Likewise, the decoder $\mathcal{D}$ becomes the partial isometry $U^{\mathcal{D}}_{B \rightarrow \widehat{M} E_{\mathcal{D}}}$, with an ``environment'' system $E_{\mathcal{D}}$. See Figure \ref{fig:classical-problem-purified} for an illustration of the purified picture.

Now, we will show that if dequantizing holds, then a suitable decoder must exist. Consider the two states $\rho$ and $\widehat{\rho}$ in Figure \ref{fig:classical-problem-purified}, which are the states immediately before and immediately after the decoder, and look at the reduced states on $R$, $E$ and $X$. Since these subsystems are untouched by the decoder, we have that $\rho_{RXE} = \widehat{\rho}_{RXE}$, and $\rho_{RXEB}$ and $\widehat{\rho}_{RXE \widehat{M} E_{\mathcal{D}}}$ are both purifications of this state. The decoder is then simply the partial isometry that relates them and which is guaranteed to exist by the unitary equivalence of purifications. Hence, as long as $\rho_{RXE}$ is of the right form, we know that a suitable decoder must exist.

We must now find out what this ``right form'' is. Note that since the encoder is classical, and the channel is CQ, one can show that $\rho_{RXE \widehat{M} E_{\mathcal{D}}}$ must have the form
\[ \ket{\rho}_{RXE \widehat{M} E_{\mathcal{D}}} = \sum_i \sqrt{\lambda_i} \ket{i}_R \otimes \ket{\pi(i)}_X \otimes \ket{\psi^{i}}_{E \widehat{M} E_{\mathcal{D}}}, \]
for some set of states $\{ \ket{\psi^i} \}$ and some permutation $\pi$. Furthermore, we know that the error probability must be low; this means that $\ket{\rho}$ must be close to a state of the form
\[ \ket{\xi}_{RXE \widehat{M} E_{\mathcal{D}}} = \sum_i \sqrt{\lambda_i} \ket{i}_R \otimes \ket{\pi(i)}_X \otimes \ket{i}_{\widehat{M}} \otimes \ket{\theta^{i}}_{E E_{\mathcal{D}}}, \]
where here the decoder output $\widehat{M}$ is perfectly correlated with the message in $R$ (and $\{ \ket{\theta^i} \}$ is some set of states). Tracing out $\widehat{M} E_{\mathcal{D}}$ in $\xi$, we get
\[ \xi_{RXE} = \sum_i \lambda_i \proj{i \pi(i)}{i \pi(i)}_{RX} \otimes \theta^i_E. \]
Note that this $\xi$ has only classical correlations between the three systems --- this is the ``right form'' that we need for the channel output.

Hence, we have reduced the problem to finding an encoder that ensures that the output of the channel has almost only classical correlations, and this is precisely what the dequantizing theorem does.

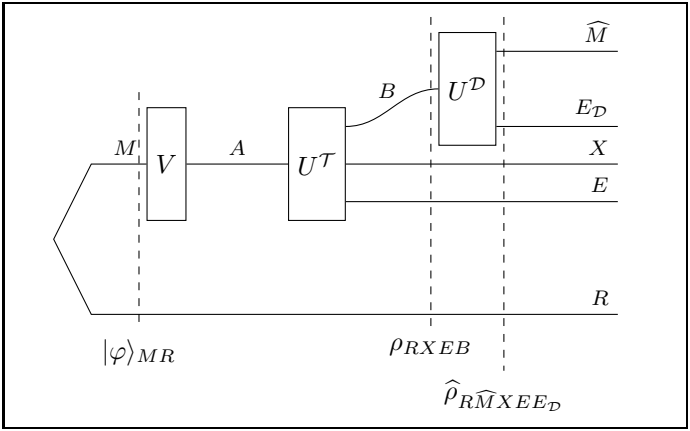
\begin{figure}
	\centering
		\def\Wall{8}
		\begin{tikzpicture}[auto]
		\tikzstyle{gate} = [fill=white, draw]
		\tikzstyle{syslabel} = [font=\footnotesize]

		\node[gate,minimum height=1.5cm] (encoder) at (2,3) {$V$};
		\node[gate,minimum height=1.5cm] (channel) at (4,3) {$U^{\mathcal{T}}$};
		\node[gate,minimum height=1.5cm] (decoder) at (6,4) {$U^{\mathcal{D}}$};

		\draw (encoder) -- ++(-1,0) -- ++(-.5,-1) -- ++(.5,-1) coordinate (bottomline) -- (\Wall,1);
		\draw (encoder) to node[syslabel] {$A$} (channel);
		\draw (channel.east)+(0,.5) to [out=0,in=180] node[syslabel,xshift=5] {$B$} (decoder);
		\draw (channel.east) -- (\Wall,3) node[syslabel,above left] {$X$};
		\draw (channel.east)+(0,-.5) -- (\Wall,2.5) node[syslabel,above left] {$E$};
		\draw (decoder.east)+(0,.5) -- (\Wall,4.5) node[syslabel, above left] {$\widehat{M}$};
		\draw (decoder.east)+(0,-.5) -- (\Wall,3.5) node[syslabel, above left] {$E_{\mathcal{D}}$};

		\node[syslabel,above left] at (encoder.west) {$M$};
		\node[syslabel,above left] at (\Wall,1) {$R$};

		\draw[dashed] ([xshift=-.1cm,yshift=.2cm] encoder.west |- encoder.north) -- ([xshift=-.1cm,yshift=-.2cm] encoder.west |- bottomline) node[below] {$\ket{\varphi}_{MR}$};
		\draw[dashed] ([xshift=-.1cm,yshift=.2cm] decoder.west |- decoder.north) -- ([xshift=-.1cm,yshift=-.2cm] decoder.west |- bottomline) node[below] {$\rho_{R X EB}$};
		\draw[dashed] ([xshift=.1cm,yshift=.2cm] decoder.east |- decoder.north) -- ([xshift=.1cm,yshift=-.75cm] decoder.east |- bottomline) node[below] {$\widehat{\rho}_{R \widehat{M}XE E_{\mathcal{D}}}$};

		\end{tikzpicture}
	     \caption{Diagram illustrating the completely purified scenario. Since $\cT$ is a CQ-channel the corresponding environment can be split up in two parts $X$ and $E$. }
	\label{fig:classical-problem-purified}
\end{figure}

\subsection{A one-shot classical coding theorem}
We now put the pieces together and derive a coding theorem based on the argument in the previous subsection. We need to obtain an encoding operation with the property that $\ptrace{B}{\bar{\cT}\circ\cE(\varphi_{MR})} \approx \sum_i \lambda_i \proj{i\pi(i)}{i\pi(i)}_{RX} \otimes \theta^{i}_E$. On the other hand, note that by sending the \emph{classically} correlated state $\varphi^{\cl}_{MR}$ through the channel (as opposed to $\ket{\varphi}_{MR}$), we automatically get a state of this form, regardless of the encoder. Our strategy will therefore be to show that
\begin{equation}
	\bar{\cT}\left(\mathcal{E}(\varphi_{MR})\right) \approx  {\bar{\cT}\left(\mathcal{E}(\varphi_{MR}^{\cl})\right)}.   \label{eqn:desired-condition}
\end{equation}

We are now in a position to use the dequantizing theorem. The encoder is constructed as follows: we first embed the message $M$ into the input space of the channel $A$. (We could denote this using a partial isometry, but to avoid cluttering the notation we will simply consider $\mathcal{H}_M$ to be a subspace of $\mathcal{H}_A$ from now on.) We then apply a permutation on the basis elements of $A$, as in Theorem \ref{thm:maintheorem}. We will then show that, if we average over the choice of permutations, this strategy works. It then follows that a suitable permutation exists.

Applying Theorem \ref{thm:maintheorem} to the scenario at hand, we get that there is a permutation operator such that
\begin{multline*}
\Norm{\bar{\mathcal{T}}(\idx{P}{A} (\varphi_{MR}-\varphi_{MR}^{\cl}) \idx{P}{A}^\dagger )}{1}\\
\leq\sqrt{\frac{1}{\idx{d}{A}-1}\:2^{-\chmineps{A}{EX}{\omega}-\chmin{M}{R}{\varphi}}}+8\varepsilon.
\end{multline*}
This gives precise bounds for (\ref{eqn:desired-condition}) above. We now get the decoder using Corollary~\ref{uhlcor}: there exists a TPCPM $\cD_{B\rightarrow\widehat{M}}$ such that
\begin{multline*}
	\Norm{\cD\left( U_\mathcal{T} \idx{P}{A} (\varphi_{MR}-\varphi_{MR}^{\cl}) \idx{P}{A}^\dagger U_{\cT}^{\dagger} \right)}{1}\\
\leq2\sqrt{\sqrt{\frac{1}{\idx{d}{A}-1}\:2^{-\chmineps{A}{EX}{\omega}-\chmin{M}{R}{\varphi}}}+8\varepsilon}.
\end{multline*}
Tracing out the systems $X$ and $E$, we get
\begin{multline*}
\Norm{\cD\circ\mathcal{T}(\idx{P}{A} \varphi_{MR} \idx{P}{A}^\dagger)-\varphi_{\widehat{M}R}^{\cl}}{1}\\
\leq2\sqrt{\sqrt{\frac{1}{\idx{d}{A}-1}\:2^{-\chmineps{A}{EX}{\omega}-\chmin{M}{R}{\varphi}}}+8\varepsilon}.
\end{multline*}

By the duality of the smooth min- and max-entropies~\cite{duality-min-max-entropy}, we have $-H_{\min}(M|R)_{\varphi} = H_{\max}(M)_{\varphi}$ and $-H_{\min}^{\varepsilon}(A|EX)_{\omega} = H_{\max}^{\varepsilon}(A|B)_{\omega}$, yielding the following theorem:
\begin{theorem} \label{oneshotHSW}
	Let $\mathcal{T}_{A \rightarrow B}$ be a quantum channel with Choi-Jamiołkowski representation $\omega_{AB}$, and let $\ket{\varphi}_{MR} = \sum_i \sqrt{\lambda_i} \ket{ii}_{MR}$, where $\lambda_i$ is a probability distribution. Then, we have that there exists a permutation $P_A$ on the basis elements of $A$ such that
\begin{multline}
\Norm{\cD\circ\mathcal{T}(\idx{P}{A} \varphi_{MR} \idx{P}{A}^\dagger)-\varphi_{\widehat{M}R}^{\cl}}{1}\\
\leq2\sqrt{\sqrt{\frac{1}{\idx{d}{A}-1}\:2^{H_{\max}^{\varepsilon}(A|B)_{\omega} + H_{\max}(M)_{\varphi}}}+8\varepsilon}. \label{eqn:coding-theorem-bound}
\end{multline}
\end{theorem}

Hence, if the right-hand side of the above inequality is small enough, then the scheme succeeds. We formulate this fact as the following corollary of the preceding theorem:

\begin{corollary}\label{cor:oneshot-rates}
Let $\cTAE$ be a quantum channel with Choi-Jamiołkowski representation $\idx{\omega}{AB}$, and let $\mathcal{M}$ be a set of messages with $\abs{\mathcal{M}} \leq \idx{d}{A}$ and $p$ a probability distribution on $\mathcal{M}$. Then, there exists an encoder and a decoder for $\cTAE$ with error probability $p_e \geq 0$ if
\begin{multline*}
  H_{\max}(M)_{p} \leq\\ \log \idx{d}{A} - H_{\max}^{\eps}(A|B)_{\omega} - 1 + 2 \log \Big( p_e^2 - 8\eps \Big).
\end{multline*}
for some constant $\eps$, $0 \leq \eps \leq \frac{p_e^2}{8}$.
\end{corollary}
\begin{IEEEproof}
	The abvoe inequality ensures that the right-hand side of (\ref{eqn:coding-theorem-bound}) is at most $2p_e$.
\end{IEEEproof}



\subsection{The i.i.d scenario and the HSW theorem}\label{sec:iid}
Applying the above to a channel of the form $\mathcal{T}_{A \rightarrow B}^{\otimes n}$ allows us to easily recover the HSW theorem. Recall that the HSW theorem states that there exists a family of codes for $\mathcal{T}$ with increasing block length $n$ with a vanishing error probability as $n \rightarrow \infty$ as long as the rate $Q$ is less than $I(X,B)_{\tau}$, where $\tau$ is a state of the form $\tau_{XB} = \sum_x p_x \proj{x}{x}_X \otimes \mathcal{T}(\sigma^x_A)$, where $p_x$ forms a probability distribution. (The rate of a code for $n$ uses of a channel is $\frac{1}{n} \log K$, where $K$ is the number of possible messages sent.) The only challenge facing us when attempting to prove this is to relate the quantity $\log d_A$ to $nH(X)_{\tau}$. We do this using the idea of \emph{types}, explained very briefly in Appendix \ref{sec:method-of-types}. The result is the following theorem:

\begin{theorem}[Holevo \cite{holevo98}, Schumacher-Westmoreland \cite{SW97}]
	Let $\mathcal{T}_{A \rightarrow B}$ be a CQ channel, let $q$ be a probability distribution over the set $\mfX$, and let $\tau_{XB} := \sum_{x \in \mfX} q(x) \proj{x}{x}_X \otimes \mathcal{T}(\sigma_A^x)$, where $\{ \sigma_A^x : x \in \mfX \}$ is a set of states on $A$. Then, there exists a family of codes for $\mathcal{T}^{\otimes n}$ whose rate approaches $I(X;B)_{\tau}$.
\end{theorem}
\begin{proof}
	Let $p \in \mathcal{P}_n(\mfX)$ be the most likely type under the distribution $q^n$, and let $A'$ be a system with $\mathcal{H}_{A'} \subseteq \mathcal{H}_A^{\otimes n}$ of dimension $|t(p)|$ defined as
	\[ \mathcal{H}_{A'} := \Span \{ \ket{\vec{x}}_{A^n} : \vec{x} \in t(p) \}. \]
	Now, consider the channel $\mathcal{T}'_{A' \rightarrow B^{\otimes n}}$ defined as follows:
	\[ \mathcal{T}'(\xi_{A'}) = \sum_{\vec{x} \in t(p)} \bra{\vec{x}} \xi \ket{\vec{x}} \mathcal{T}^{\otimes n}(\sigma_{\vec{x}}), \]
	where $\sigma_{\vec{x}} = \sigma_{x_1} \otimes \dots \otimes \sigma_{x_n}$. Furthermore, let $\omega'_{A'B^n}$ be the Choi-Jamiołkowski state of $\mathcal{T}'$, and let us apply Corollary \ref{cor:oneshot-rates} to the uniform distribution over some message set $\mathcal{M}$ and channel $\mathcal{T}'$.  We get that there exists an encoder and a decoder such that 
\begin{align*}
	\log |\mathcal{M}| &\leqslant \log d_{A'} - H_{\max}^{\varepsilon'}(A'|B^n)_{\omega} - o(n)\\
	&= \log |t(p)| - H_{\max}^{\varepsilon'}(A'|B^n)_{\omega} - o(n)\\
\end{align*}
From the properties of types, we have that $|t(p)| \geqslant |\mathcal{P}_n(\mfX)|^{-1} 2^{nH(X)_{\tau}}$. Also, note that 
\[ \omega_{A'B^n} = \Pi_{t(p)} \tau_{XB}^{\otimes n} \Pi_{t(p)}/\tr[\Pi_{t(p)} \tau_{XB}^{\otimes n}]. \]
Using  Lemma \ref{lem:hmax-proj}, we get that 
\[ H_{\max}^{\varepsilon}(A'|B^n)_{\omega'} \leqslant H_{\max}^{\varepsilon}(X^n|B^n)_{\tau^{\otimes n}} + \log \tr[\Pi_{t(p)} \rho_{XB}^{\otimes n}]. \]
Hence, we now have that
\begin{multline*}
	\log |\mathcal{M}| \leqslant \log|t(p)|\\
- H_{\max}^{\varepsilon}(X^n|B^n)_{\tau^{\otimes n}} - \log\tr[\Pi_{t(p)} \tau_{XB}^{\otimes n}] - o(n).
\end{multline*}
Choosing $\log |\mathcal{M}| = nQ$ for a transmission rate of $Q$ and using the above bound on $|t(p)|$ and the fact that the most likely type $p$ satisfies $\tr[\Pi_{t(p)} \tau_{XB}^{\otimes n}] \geqslant |\mathcal{P}_n(\mfX)|^{-1}$, this bound becomes:
\begin{equation*}
	Q \leqslant H(X)_{\tau} - \frac{1}{n} H_{\max}^{\varepsilon}(X^n|B^n)_{\tau^{\otimes n}} - o(1).
\end{equation*}

Taking the limit as $n \rightarrow \infty$, we get that this bound is satisfied whenever
\[ Q < H(X)_{\tau} - H(X|B)_{\tau} = I(X;B)_{\tau}, \]
where we have used the fully quantum asymptotic equipartition property of \cite{tcr08} to bound the $H_{\max}$ term above. Since this is true for any $\varepsilon > 0$, the theorem holds.
\end{proof}

\section{Conclusion and further work}\label{sec:discussion}
In this article, we show that it is possible to derive direct bounds for the capacity of classical-quantum channels using decoupling-like techniques, therefore adding the transmission of classical data to the list of problems that are amenable to the decoupling approach to coding. Our derivation also naturally leads to bounds in the one-shot setting, where the channel is only used once and we allow a finite error probability.

We want to emphasize, however, that the bounds resulting from our calculation are somewhat weaker than the best known one-shot direct bounds, found for example in Mosonyi and Datta~\cite{md09} and Wang and Renner~\cite{wr10}. Furthermore, our one-shot result only applies to uniform inputs of the channel and the method of types is needed to shape the input into this form in order to achieve the HSW capacity. The latter weakness could potentially be overcome inside the decoupling framework, using an analogue of Theorem~3.14 in~\cite{fred-these}.

\appendices

\section{Technical facts about the smooth entropies}\label{A}

Here we establish some useful properties of the (smooth) min-entropy of classically coherent states. In particular for a state $\rho_{XX'AB}$ that is classically coherent between $X$ and $X\rq{}$ we show that the state $\sigma_{X\rq{}B}$ that optimizes
\begin{multline*}
H_{\min}(XA|X'B)_{\rho}\\
=\max_{\sigma_{X\rq{}B}}\sup \big\{ \lambda \in \mathbb{R} : \tilde{\rho}_{XX'AB} \leq 2^{-\lambda} \id_{XA} \otimes \sigma_{X'B} \big\}
\end{multline*}
can be chosen to have CQ-structure. Furthermore we show how the min-entropy of a classically coherent state behaves under smoothing. The following lemma is a direct consequence of the results obtained in~\cite{duality-min-max-entropy}.
\begin{lemma}\label{marcoslemma}
  Let $\rho_{XX'AB}$ be coherent classical on $X$ and $X'$. Then, there exists a state $\tilde{\rho}_{XX'AB} \in \mathcal{B}^{\eps}(\rho_{XX'AB})$ that is coherent classical on $X$ and $X'$ and a state $\sigma_{X'B} \in \mathcal{S}_{\leq}(\mathcal{H}_{X'B})$ that is classical on $X'$ such that
  \begin{multline*}
    H_{\min}^{\eps}(XA|X'B)_{\rho}\\ 
      = \sup \big\{ \lambda \in \mathbb{R} : \tilde{\rho}_{XX'AB} \leq 2^{-\lambda} \id_{XA} \otimes \sigma_{X'B} \big\} .
  \end{multline*}
\end{lemma}

\begin{IEEEproof}
  Let $\widehat{\rho}_{XX'AB} \in \mathcal{B}^{\eps}(\rho_{XX'AB})$ be a state that maximizes the smooth min-entropy, namely it satisfies $H_{\min}^{\eps}(XA|X'B)_{\rho}$ $=$ $H_{\min}(XA|X'B)_{\widehat{\rho}}$. Then, the state $\tilde{\rho}_{XX'AB} = P_{XX'} \bar{\rho}_{XX'AB} P_{XX'}$ satisfies the criteria.
  
  First, note that we have $P(\tilde{\rho}_{XX'AB}, \rho_{XX'AB}) \leq P(\widehat{\rho}_{XX'AB}, \rho_{XX'AB}) \leq \eps$ due to the monotonicity of the purified distance under projections~\cite{duality-min-max-entropy}.
  Second, by definition of the smooth min-entropy, 
  there exists an operator $\widehat{\sigma}_{X'B}$ such that, for $\lambda = 
  H_{\min}^{\eps}(XA|X'B)$, we have
  \begin{align*}
    \widehat{\rho}_{XX'AB} \leq 2^{-\lambda}\, \id_{XA} \otimes \widehat{\sigma}_{X'B} \,.
  \end{align*} 

Thus,
\begin{align}
    \rhot_{XX'AB} &\leq 2^{-\lambda}\, P_{XX'} \big( \id_{XA} \otimes \widehat{\sigma}_{X'B} \big) P_{XX'} \nonumber\\
    &= 2^{-\lambda} \sum_x \proj{x}{x}_{X} \otimes \id_{A} 
      \otimes \proj{x}{x}_{X'} \otimes \bracket{x}{\widehat{\sigma}_{X'B}}{x} \nonumber\\
    &\leq 2^{-\lambda}\, \id_{XA} \otimes \underbrace{\sum_x \proj{x}{x}_{X'} \kron \bracket{x}{\widehat{\sigma}_{X'B}}{x}}_{=: \sigma_{X'B}}
    \label{eq:marcoslemma} \,.
  \end{align}
  Finally, we note that $\tr(\sigma_{X'B}) \leq 1$ and, thus, Eq.~\eqref{eq:marcoslemma} implies that $H_{\min}^{\eps}(XA|X'B)_{\rho} \geq \lambda$, which concludes the proof.
\end{IEEEproof}

\begin{lemma}\label{lem:hmax-proj}
	Let $\rho_{AB} \in \subnormstates{\cH_{AB}}$, let $\Pi_A$ be an operator such that $0 \leqslant \Pi_A \leqslant \ident_A$, and let $\varepsilon \geqslant 0$. Furthermore, let $\rho'_{AB} := \Pi_A \rho_{AB} \Pi_A$. Then, $H^{\varepsilon}_{\max}(A|B)_{\rho'} \leqslant H^{\varepsilon}_{\max}(A|B)_{\rho}$.
\end{lemma}
\begin{IEEEproof}
	Let $\tilde{\rho}_{AB} \in \mathcal{B}_{\varepsilon}(\rho)$ and $\sigma_B$ be such that
	\[ 2^{H_{\max}^{\varepsilon}(A|B)_{\rho}} = F(\tilde{\rho}_{AB}, \ident_A \otimes \sigma_B)^2. \]
	Let $\tilde{\rho}' := \Pi \tilde{\rho} \Pi \in \mathcal{B}_{\varepsilon}(\rho')$, and let $\omega_B$ be such that
	\[ 2^{H_{\max}(A|B)_{\tilde{\rho}'}} = F(\tilde{\rho}'_{AB}, \ident_A \otimes \omega_B)^2. \]
	Then, we immediately have that
	\begin{align*}
		2^{H^{\varepsilon}_{\max}(A|B)_{\rho}} &= F(\tilde{\rho}_{AB}, \ident_A \otimes \sigma_B)^2\\
		&\geqslant F(\tilde{\rho}_{AB}, \ident_A \otimes \omega_B)^2\\
		&\geqslant F(\tilde{\rho}_{AB}, \Pi_A^2 \otimes \omega_B)^2\\
		&= \tr\left[ \sqrt{(\Pi_A \otimes \sqrt{\omega_B}) \tilde{\rho}_{AB} (\Pi_A \otimes \sqrt{\omega_B})} \right]^2\\
		&= \tr\left[ \sqrt{(\ident_A \otimes \sqrt{\omega_B}) \Pi_A \tilde{\rho}_{AB} \Pi_A (\ident_A \otimes \sqrt{\omega_B})} \right]^2\\
		&= F(\Pi_A \tilde{\rho}_{AB} \Pi_A, \ident_A \otimes \omega_B)^2\\
		&= F(\tilde{\rho}_{AB}', \ident_A \otimes \omega_B)^2\\
		&= 2^{H_{\max}(A|B)_{\tilde{\rho}'}}\\
		&\geqslant 2^{H^{\varepsilon}_{\max}(A|B)_{\rho'}}.
	\end{align*}
	Taking logarithms then yields the lemma.
\end{IEEEproof}

\section{The method of types}\label{sec:method-of-types}
The ``method of types'' is a technique that is used extensively in classical information theory and that we need here to make statements about discrete memoryless channels. For a complete introduction to this method, we refer the reader to \cite{method-of-types}; we will only give here the facts needed for our paper. The basic idea goes as follows. Let $\mfX$ be a finite set, and let $\vec{x} = x_1 \dots x_n \in \mfX^n$ be a sequence of $n$ symbols from $\mfX$. For any $x \in \mfX$, let $p_{\vec{x}}(x)$ be the relative frequency of the symbol $x$ in $\vec{x}$ (i.e. the number of occurences of $x$ in $\vec{x}$ divided by $n$). We call the distribution $p_{\vec{x}}$ the \emph{type} of $\vec{x}$, and, given a type $p$, we define $t(p)$ to be the set of all sequences of type $p$. Furthermore, we define $\mathcal{P}_n(\mfX)$ to be the set of all types over $\mfX^n$.

We now list some basic properties of types:
\begin{itemize}
	\item $|\mathcal{P}_n(\mfX)| = { n + |\mfX| - 1 \choose |\mfX| - 1}$.
	\item For any type $p \in \mathcal{P}_n(\mfX)$, we have that
		\[ |\mathcal{P}_n(\mfX)|^{-1} 2^{nH(p)} \leqslant |t(p)| \leqslant 2^{nH(p)}. \]
	\item For any type $p \in \mathcal{P}_n(\mfX)$ and any probability distribution $q$ over $\mfX$, we have that
		\[ |\mathcal{P}_n(\mfX)|^{-1} 2^{-nD(p \| q)} \leqslant \sum_{\vec{x} \in t(p)} q^n(\vec{x}) \leqslant 2^{-nD(p \| q)}. \]
	\item For any probability distribution $q$ and any $n$, the most likely type $p$ has total probability
		\[ \sum_{\vec{x} \in t(p)} q^n(\vec{x}) \geqslant |\mathcal{P}_n(\mfX)|^{-1}. \]
\end{itemize}
Note that $|\mathcal{P}_n(\mfX)|$ is polynomial in $n$ and becomes negligible in most expressions involving exponentials of entropies.

To use these concepts in quantum information, we will define \emph{type projectors}. Let $X$ be a $|\mfX|$-dimensional quantum system, with a basis vector $\ket{x}$ for each $x \in \mfX$. Let $p \in \mathcal{P}_n(\mfX)$; we define the type projector $\Pi_{t(p)}$ as
\[ \Pi_{t(p)} = \sum_{\vec{x} \in t(p)} \proj{\vec{x}}{\vec{x}}, \]
where $\ket{\vec{x}} = \ket{x_1} \otimes \dots \otimes \ket{x_n}$ for $\vec{x} = x_1 \dots x_n$.

\section*{Acknowledgments}
We thank Renato Renner for discussions. FD acknowledges support from the Swiss National Science Foundation (grants PP00P2-128455), the National Centre of Competence in Research ``Quantum Science and Technology'') and the German Science Foundation (grants \mbox{CH~843/1-1} and \mbox{CH~843/2-1}). OS acknowledges funding by the Elite Network of Bavaria (ENB) project QCCC.
MT acknowledges support from the National Research Foundation (Singapore), and the Ministry of Education (Singapore).

\bibliographystyle{abbrv}
\bibliography{refs,big}

\end{document}